\documentclass[12pt]{article}
\usepackage{amsmath, amsthm}
\usepackage{amssymb}
\usepackage{url}

\newtheorem{theorem}{Theorem}[section]
\newtheorem{proposition}[theorem]{Proposition}
\newtheorem{corollary}[theorem]{Corollary}
\newtheorem{lemma}[theorem]{Lemma}
\newtheorem{conjecture}[theorem]{Conjecture}

\title{Palindromic Length of Words with Many Periodic Palindromes}

\author{Josef Rukavicka\thanks{Department of Mathematics, Faculty of Nuclear Sciences and Physical Engineering, Czech Technical University in Prague, Czech Republic
(josef.rukavicka@seznam.cz).}}

\theoremstyle{remark}

\newtheorem{remark}[theorem]{Remark}

\DeclareMathOperator{\Factor}{Fac}
\DeclareMathOperator{\Pal}{Pal}

\DeclareMathOperator{\PalPrefix}{PalPrf}
\DeclareMathOperator{\PL}{PL}
\DeclareMathOperator{\MPF}{MPF} %minimal palindromic factorization
\DeclareMathOperator{\Prefix}{Prf}
\DeclareMathOperator{\Suffix}{Suf}

\DeclareMathOperator{\Alphabet}{A}
\DeclareMathOperator{\MinPer}{MinPer}
\DeclareMathOperator{\Period}{Period} %Period

\date{\small{May 03, 2020}}

\begin{document}
\maketitle

\begin{abstract}
The palindromic length $\PL(v)$ of a finite word $v$ is the minimal number of palindromes whose concatenation is equal to $v$. In 2013, Frid, Puzynina, and Zamboni conjectured that: If $w$ is an infinite word and $k$ is an integer such that $\PL(u)\leq k$ for every factor $u$ of $w$ then $w$ is ultimately periodic. 

Suppose that $w$ is an infinite word and $k$ is an integer such $\PL(u)\leq k$ for every factor $u$ of $w$. Let $\Omega(w,k)$ be the set of all factors $u$ of $w$ that have more than $\sqrt[k]{k^{-1}\vert u\vert}$ palindromic prefixes. We show that $\Omega(w,k)$ is an infinite set and we show that for each positive integer $j$ there are palindromes $a,b$ and a word $u\in \Omega(w,k)$ such that $(ab)^j$ is a factor of $u$ and $b$ is nonempty. Note that $(ab)^j$ is a periodic word and $(ab)^ia$ is a palindrome for each $i\leq j$.
These results justify the following question: What is the palindromic length of a concatenation of a suffix of $b$ and a periodic word $(ab)^j$ with ``many'' periodic palindromes?

It is known that $\lvert\PL(uv)-\PL(u)\rvert\leq \PL(v)$, where $u$ and $v$ are nonempty words. The main result of our article shows that if $a,b$ are palindromes, $b$ is nonempty, $u$ is a nonempty suffix of $b$, $\vert ab\vert$ is the minimal period of $aba$, and $j$ is a positive integer with $j\geq3\PL(u)$ then $\PL(u(ab)^j)-\PL(u)\geq 0$.
\end{abstract}

\section{Introduction}

In 2013, Frid, Puzynina, and Zamboni introduced a \emph{palindromic length} of a finite word \cite{FrPuZa}. Recall that the word $u=x_1x_2\dots x_n$ of length $n$ is called a \emph{palindrome} if $x_1x_2\dots x_n=x_n\dots x_2x_1$, where $x_i$ are letters and $i\in \{1,2,\dots,n\}$. The palindromic length $\PL(u)$ of the word $u$ is defined as the minimal number $k$ such that $u=u_1u_2\dots u_k$ and $u_j$ are palindromes, where $j\in \{1,2,\dots,k\}$; note that the palindromes $u_j$ are not necessarily distinct. Let $\epsilon$ denote the empty word. We define that $\PL(\epsilon)=0$. 

In general, the factorization of a finite word into the minimal number of palindromes is not unique; for example $\PL(011001)=3$ and the word $011001$ can be factorized in two ways: $011001=(0110)(0)(1)=(0)(1)(1001)$.

The authors of \cite{FrPuZa} conjectured that: 
\begin{conjecture}
\label{tue8eiru883}
If $w$ is an infinite word and $P$ is an integer such that $\PL(u)\leq P$ for every factor $u$ of $w$ then $w$ is ultimately periodic. 
\end{conjecture}
So far, Conjecture \ref{tue8eiru883} remains open.
We call an infinite word that satisfies the condition from Conjecture \ref{tue8eiru883} a word with a \emph{bounded palindromic length}. Note that there are infinite periodic words that do not have a bounded palindromic length; for example $(012)^{\infty}$. Hence the converse of Conjecture \ref{tue8eiru883} does not hold.

In \cite{FrPuZa} the conjecture was proved for infinite words that are $k$-power free for some positive integer $k$.  It follows that if $w$ is an infinite word with a bounded palindromic length, then for each positive integer $j$ there is a nonempty factor $r$ such that $r^j$ is a factor of $w$.  

%In addition, the conjecture was proved for words that satisfy the so called $(k,l)$-condition \cite{FrPuZa}.

In \cite{10.1007/978-3-319-66396-8_19}, another variation of Conjecture \ref{tue8eiru883} was considered:
\begin{conjecture}
\label{di88ejdui33}
Every aperiodic (not ultimately periodic) infinite word has prefixes of arbitrarily high palindromic length.
\end{conjecture}
In \cite{10.1007/978-3-319-66396-8_19}, the author proved that Conjecture \ref{tue8eiru883} and Conjecture \ref{di88ejdui33} are equivalent. More precisely, it was proved that if every prefix of an infinite word $w$ is a concatenation of at most $n$ palindromes then every factor of $w$ is a concatenation of at most $2n$ palindromes. It follows that Conjecture \ref{di88ejdui33} remains also open.

In \cite{FRID2018202} Conjecture \ref{tue8eiru883} and Conjecture \ref{di88ejdui33} have been proved for all Sturmian words. The properties of the palindromic length of Sturmian words have been investigated also in \cite{10.1007/978-3-030-24886-4_18}. In \cite{AMBROZ201974}, the authors study the palindromic length of factors of fixed points of primitive morphisms. In \cite{10.1007/978-3-030-24886-4_17}, the lower bounds for the palindromic length of prefixes of infinite words can be found.

In \cite{BucMichGreedy2018}, a left and right greedy palindromic length have been introduced as a variant to the palindromic length. It is shown that if the left (or right) greedy palindromic lengths of prefixes of an infinite word $w$ is bounded, then $w$ is ultimately periodic.  

In addition, algorithms for computing the palindromic length were researched \cite{borozdin_et_al:LIPIcs:2017:7338}, \cite{FICI201441}, \cite{RuSh15}. In \cite{RuSh15}, the authors present a linear time online algorithm for computing the palindromic length.

In the current paper we investigate infinite words with a bounded palindromic length. Let $k$ be a positive integer, let $w$ be an infinite word such that $k\geq\PL(t)$ for every factor $t$ of $w$, and let $\Omega(w,k)$ be the set of all factors $u$ of $w$ that have more than $\sqrt[k]{k^{-1}\vert u\vert}$ palindromic prefixes. We show that $\Omega(w,k)$ is an infinite set and we show that for each positive integer $j$ there are palindromes $a,b$ and a word $u\in \Omega(w,k)$ such that $(ab)^j$ is a factor of $u$ and $b$ is nonempty. Note that $(ab)^j$ is a periodic word and $(ab)^ia$ is a palindrome for each $i\leq j$. In this sense we can consider that $w$ has infinitely many periodic palindromes with an arbitrarily high exponent $j$. 

The existence of infinitely many periodic palindromes in $w$ is not surprising. It can be deduced also from the result in \cite{FrPuZa}, which says, as mentioned above, that if $w$ is an infinite word with a bounded palindromic length, then for each positive integer $j$ there is a nonempty factor $r$ such that $r^j$ is a factor of $w$. 

These results justify the following question: What is the palindromic length of a concatenation of a suffix of $b$ and a periodic word $(ab)^j$ with ``many'' periodic palindromes?

It is known that if $u,v$ are nonempty words then $\lvert\PL(uv)-\PL(u)\rvert\leq \PL(v)$ \cite{10.1007/978-3-319-66396-8_19}. Less formally said, it means that by concatenating a word $v$ to a word $u$ the change of the palindromic length is at most equal to the palindromic length of $v$. The main result of our article shows that if $a,b$ are palindromes, $b$ is nonempty, $u$ is a nonempty suffix of $b$, $\vert ab\vert$ is the minimal period of $aba$, and $j$ is a positive integer with $j\geq3\PL(u)$ then $\PL(u(ab)^j)-\PL(u)\geq 0$. 

The results of our article should shed some light on infinite words for which Conjecture  \ref{tue8eiru883} and Conjecture \ref{di88ejdui33} remain open.

\section{Preliminaries}

Let $\mathbb{N}$ denote the set of all positive integers, let $\mathbb{N}_0=\mathbb{N}\cup\{0\}$ denote the set of all nonnegative integers, let $\mathbb{R}$ denote the set of all real numbers, and let $\mathbb{R}^+$ denote the set of all positive real numbers. 

Let $\Alphabet$ denote a finite alphabet with $\vert \Alphabet\vert\geq 2$ letters. Let $\Alphabet^+$ denote the set of all finite nonempty words over the alphabet $\Alphabet$ and let $\Alphabet^*=\Alphabet^+\cup \{\epsilon\}$; recall that $\epsilon$ denotes the empty word. Let $\Alphabet^{\mathbb{N}}$ denote the set of all right infinite words.

Let $n\in \mathbb{N}$ and let $w=w_1w_2\dots w_n\in \Alphabet^*$, where $w_i\in \Alphabet$ and $i\in\{1,2,\dots,n\}$. We denote by $w[i,j]=w_iw_{i+1}\dots w_j$ the factor of $w$ starting at position $i\in \mathbb{N}$ and ending at position $j\in\mathbb{N}$, where $i,j\in \mathbb{N}$ and $i\leq j\leq n$

Let $w=w_1w_2\dots \in \Alphabet^{\mathbb{N}}$, where $w_i\in \Alphabet$ and $i\in\{1,2,\dots\}$. We denote by $w[i,j]=w_iw_{i+1}\dots w_j$ the factor of $w$ starting at position $i\in \mathbb{N}$ and ending at position $j\in\mathbb{N}$, where $i,j\in \mathbb{N}$ and $i\leq j$.

We call the word $v\in \Alphabet^*$ a \emph{factor} of the word $w\in \Alphabet^*\cup\Alphabet^{\mathbb{N}}$ if there are words $a\in \Alphabet^*$ and $b\in \Alphabet^*\cup\Alphabet^{\mathbb{N}}$ such that $w=avb$. Given a word $w\in \Alphabet^*\cup\Alphabet^{\mathbb{N}}$, we denote by $\Factor(w)$ the set of all factors of $w$. It follows that $\epsilon\in \Factor(w)$ and if $w\in \Alphabet^*$ then also $w\in \Factor(w)$. 

We call the word $v\in \Alphabet^*$ a \emph{prefix} of the word $w\in \Alphabet^*\cup\Alphabet^{\mathbb{N}}$ if there is $t\in \Alphabet^*\cup\Alphabet^{\mathbb{N}}$ such that $w=vt$. Given a word $w\in \Alphabet^*\cup\Alphabet^{\mathbb{N}}$, we denote by $\Prefix(w)$ the set of all prefixes of $w$. It follows that $\epsilon\in \Prefix(w)$ and if $w\in \Alphabet^*$ then also $w\in \Prefix(w)$. 

We call the word $v\in \Alphabet^*$ a \emph{suffix} of the word $w\in \Alphabet^*$ if there is $t\in \Alphabet^*$ such that $w=tv$. Given a word $w\in \Alphabet^*$, we denote by $\Suffix(w)$ the set of all suffixes of $w$. It follows that $\epsilon, w\in \Suffix(w)$. 

Let $w=w_1w_2\dots w_n\in \Alphabet^+$, where $w_i\in \Alphabet$ and $i\in \{1,2,\dots, n\}$.
Let $w^R$ denote the \emph{reversal} of the word $w\in \Alphabet^+$; it means $w^R=w_nw_{n-1}\dots w_2w_1$. In addition we define that the reversal of the empty word is the empty word; formally $\epsilon^R=\epsilon$. 

Realize that $w\in \Alphabet^*$  is a \emph{palindrome} if and only if $w^R=w$. Let $\Pal\subset \Alphabet^*$ denote the set of all palindromes over the alphabet $\Alphabet$. We define that $\epsilon\in \Pal$. Let $\Pal^+=\Pal\setminus\{\epsilon\}$ be the set of all nonempty palindromes. 

Given $w\in \Alphabet^*\cup \Alphabet^{\mathbb{N}}$, let $\PalPrefix(w)=\Pal\cap\Prefix(w)$ be the set of all palindromic prefixes of $w$.

Given $w\in\Alphabet^+$, let $\MPF(w)$ denote the set of all $k$-tuples of palindromes whose concatenation is equal to $w$ and $k=\PL(w)$; formally
\[\begin{split}\MPF(w)=\{(t_1,t_2,\dots,t_k)\mid k=\PL(w)\mbox{ and }t_1t_2\dots t_k=w\\ \mbox{ and }t_1,t_2,\dots,t_k\in \Pal^+\}\mbox{.}\end{split}\]
We call a $k$-tuple $(t_1,t_2,\dots,t_k)\in \MPF(w)$ a \emph{minimal palindromic factorization} of $w$.

Let $\mathbb{Q}$ denote the set of all rational numbers. We say that the word $w\in \Alphabet^+$ is a \emph{periodic} word, if there are $\alpha\in \mathbb{Q}$, $r\in \Prefix(w)\setminus\{\epsilon\}$, and $\bar r\in \Prefix(r)\setminus\{r\}$ such that $\alpha>1$, $w=rr\dots r\bar r$, and $\frac{\vert w\vert}{\vert r\vert}=\alpha$; note that $\bar r$ is uniquely determined by $r$. We write $w=r^{\alpha}$ and the period of $w$ is equal to $\vert r\vert$. For example $12341=(1234)^{\frac{5}{4}}$ and $12341234123=(1234)^{\frac{11}{4}}$.

Given $w\in \Alphabet^+$, let \[\Period(w)=\{(r,\alpha)\mid r^{\alpha}=w\mbox{ and }r\in \Prefix(w)\setminus\{\epsilon\}\mbox{ and }\alpha\in \mathbb{Q}\mbox{ and }\alpha> 1\}\mbox{.}\]
The set $\Period(w)$ contains all couples $(r,\alpha)$ such that $r^{\alpha}=w$. Let \[\MinPer(w)=\min\{\vert r\vert\mid (r,\alpha)\in \Period(w)\}\in \mathbb{N}\mbox{.}\] The positive integer $\MinPer(w)$ is the \emph{minimal period} of the word $w$. The word $w\in \Alphabet^+$ has a period $\delta\in \mathbb{Q}$ if there is a couple $(r,\alpha)\in \Period(w)$ such that $\vert r\vert=\delta$.

We will deal a lot with periodic palindromes. The two following known lemmas will be useful for us.
\begin{lemma} (see \cite[Lemma 1]{10.1007/978-3-662-46078-8_24})
\label{tudjkdi8545}
Suppose $p$ is a period of a nonempty palindrome $w$; then there are
palindromes $a$ and $b$ such that $\vert ab\vert = p$, $b \not=\epsilon$, and $w = (ab)^*a$.
\end{lemma}
\begin{lemma} (see \cite[Lemma 2]{10.1007/978-3-662-46078-8_24})
\label{id8ieubmzmfj}
Suppose $w$ is a palindrome and $u$ is its proper suffix-palindrome or
prefix-palindrome; then the number $\vert w\vert-\vert u\vert$ is a period of $w$.
\end{lemma}

\section{Periodic palindromic factors}

We start the section with a definition of a set of real non-decreasing functions that diverge as $n$ tends towards the infinity.

Let $\Lambda$ denote the set of all functions $\phi(n)$ such that 
\begin{itemize} 
\item $\phi(n):\mathbb{N}\rightarrow\mathbb{R}$, 
\item $\phi(n)\leq \phi(n+1)$, and \item $\lim_{n\rightarrow\infty} \phi(n)=\infty$. \end{itemize}

Let $k\in \mathbb{N}$, let $\tau(n,k)=\sqrt[k]{k^{-1}n}\in \Lambda$, let $w\in \Alphabet^{\mathbb{N}}$, and let \[\Omega(w,k)=\{t\in \Factor(w)\mid \vert\PalPrefix(t)\vert\geq\tau(\vert t\vert,k)\}\mbox{.}\] The definition says that the set $\Omega(w,k)$ contains a factor $t$ of $w$ if the number of palindromic prefixes of $t$ is bigger than or equal to $\tau(\vert t\vert,k)=\sqrt[k]{k^{-1}\vert t\vert}$. 
%Let $w\in \Alphabet^{\mathbb{N}}$ be an infinite word with a bounded palindromic length. We define $\Omega(w)=\Omega(w,\sqrt[k]{k^{-1}n})$, where $k=$

The next proposition asserts that if $w$ is an infinite word with a bounded palindromic length, then the set of factors that have more than $\tau(n,k)$ palindromic prefixes is infinite, where $n$ is the length of the factor in question and $k\geq \PL(t)$ for each factor $t$ of $w$.

\begin{proposition}
\label{uyzm3m2mzo}
If $w\in \Alphabet^{\mathbb{N}}$, $k\in \mathbb{N}$ and $k\geq\max\{\PL(t)\mid t\in \Factor(w)\}$  then $\vert \Omega(w,k)\vert=\infty$.
\end{proposition}
\begin{proof}
Suppose that $\vert \Omega(w,k)\vert<\infty$ and let \[K=\max\{\vert\PalPrefix(t)\vert\mid t\in\Omega(w,k)\}\mbox{.}\] Less formally said, the value $K$ is the maximal value from the set of numbers of palindromic prefixes of factors $t$ of $w$ that have more than $\tau(\vert t\vert,k)$ palindromic prefixes. Clearly $K<\infty$, because of the assumption $\vert \Omega(w,k)\vert<\infty$. 

Let $p\in \Prefix(w)$ be the shortest prefix of $w$ such that $\tau(\vert p\vert,k)>K$.  Since $\lim_{n\rightarrow\infty}\tau(n,k)=\infty$, it is clear that such prefix $p$ exists.

To get a contradiction suppose that $\vert\PalPrefix(t)\vert\geq \tau(\vert p\vert,k)$ for some $t\in \Factor(p)$. Since $\tau(\vert t\vert,k)\leq \tau(\vert p\vert,k)$ and thus $\vert\PalPrefix(t)\vert\geq \tau(\vert t\vert,k)$, it follows that $t\in \Omega(w,k)$ and consequently $\vert \PalPrefix(t)\vert\leq K$. It is a contradiction, because $K<\tau(\vert p\vert,k)$. Hence we have that \begin{equation}\label{5hhncb2b2}\vert\PalPrefix(t)\vert< \tau(\vert p\vert,k)\mbox{ for each }t\in \Factor(p)\mbox{.}\end{equation}

Let $n,j\in\mathbb{N}$ and let \[\begin{split}\Theta(n,j)=\{(v_1,v_2,\dots, v_j)\mid v_i\in \Pal^+\mbox{ and }i\in\{1,2,\dots,j\}\mbox{ and }\\ \vert v_1v_2\dots v_j\vert\leq n\mbox{ and }v_1v_2\dots v_j\in \Prefix(w)\}\mbox{.}\end{split}\] 
The set $\Theta(n,j)$ contains $j$-tuples of nonempty palindromes whose concatenation is of length lower than or equal to $n$ and also the concatenation is a prefix of $w$.

Thus from (\ref{5hhncb2b2}) we get that \begin{equation}\label{tu884i48xbcv2}\vert \Theta(\vert p\vert,j)\vert< (\tau(\vert p\vert,k))^j\mbox{.}\end{equation} 
The equation (\ref{tu884i48xbcv2}) follows from the fact that each factor of $p$ has at most $\tau(\vert p\vert,k)$ palindromic prefixes. In consequence there are at most $(\tau(\vert p\vert,k))^j$  of $j$-tuples of palindromes.

Let $\bar \Theta(\vert p\vert,j)=\bigcup_{j>0}^k\Theta(\vert p\vert, j)$. Since $\tau(n,k)\leq\tau(n+1,k)$ we have from (\ref{tu884i48xbcv2}) that \begin{equation}\label{iu8837911jk}\vert\bar\Theta(\vert p\vert,k)\vert\leq k\vert \Theta(\vert p\vert,k)\vert< k(\tau(\vert p\vert,k))^k\leq k\left(\sqrt[k]{k^{-1}\vert p\vert}\right)^k=\vert p\vert\mbox{.}\end{equation}
The inequality (\ref{iu8837911jk}) says that the number of prefixes of $p$ having the form $v_1v_2\dots v_j$, where $j\leq k$ and $v_i\in \Pal^+$ is lower than the length of $p$. But $p$ has $\vert p\vert$ nonempty prefixes. It is a contradiction. Since $\bigcup_{r\in \Prefix(p)}\MPF(r)\subseteq\bar \Theta(\vert p\vert, k)$ we conclude that $\Omega(w,k)$ is an infinite set.
\end{proof}
\begin{remark}
In the proof of Proposition \ref{uyzm3m2mzo}, we used the idea that the number of prefixes of a word of length $n$ that are a concatenation of at most $k$ palindromes is lower than $n$. This idea  was used also in Theorem $1$ in \cite{FrPuZa}.
\end{remark}

We show that if $\Sigma$ is an infinite set of words $r$ such that the number of nonempty palindromic prefixes of $r$ grows more than $\ln{\vert r\vert}$ as $\vert r\vert$ tends towards infinity  then for each positive integer $j$ there are palindromes $a,b$ and a word $t\in \Sigma$ such that $(ab)^j$ is a prefix of $t$ and $b$ is nonempty. Realize that $(ab)^ja$ is a palindrome for each $j\in \mathbb{N}_0$. This means that $\Sigma$ contains infinitely many words that have a periodic palindromic prefix of arbitrarily high exponent $j$. 

\begin{proposition}
\label{ry73u39udh}
If $\Sigma\subseteq\Alphabet^*$, $\vert \Sigma\vert=\infty$, $\phi(n)\in \Lambda$, $\lim_{n\rightarrow\infty}\left(\phi(n)-\ln{n}\right)=\infty$, and $\vert \PalPrefix(t)\setminus\{\epsilon\}\vert\geq \phi(\vert t\vert)$ for each $t\in\Sigma$ then for each $j\in \mathbb{N}$ there are palindromes $a\in \Pal$, $b\in \Pal^+$ and a word $t\in \Sigma$ such that $(ab)^j\in \Prefix(t)$.
\end{proposition}
\begin{proof}
Given $t\in \Sigma$, let $\mu(t,i)$ be the lengths of all palindromic prefixes of $t$ such that $\mu(t,1)=1$ (a letter is a palindrome) and $\mu(t,i)<\mu(t,i+1)$, where $i\in \{1,2,\dots,h_t\}$ and $h_t=\vert \PalPrefix(t)\setminus\{\epsilon\}\vert$. The integer $h_t$ is the number of nonempty palindromic prefixes of $t$.
Let $i\in\{1,2,\dots,h_t-1\}$. It is clear that \begin{equation}\label{ty7cb2v39b}\mu(t,i+1)=\mu(t,i)\frac{\mu(t,i+1)}{\mu(t,i)}\mbox{.}\end{equation}
From (\ref{ty7cb2v39b}) we have that 
\begin{equation}\label{vbn22bxvcn0}\frac{\mu(t,h_t)}{\mu(t,h_t-1)}\frac{\mu(t,h_t-1)}{\mu(t,h_t-2)}\frac{\mu(t,h_t-2)}{\mu(t,h_t-3)}\cdots \frac{\mu(t,2)}{\mu(t,1)}=\mu(t,h_t)\leq \vert t\vert\mbox{.}\end{equation}
Suppose that there is $\alpha\in \mathbb{R}$ such that  $\alpha>1$ and for each $t\in \Sigma$ and for each $i\in \{1,2,\dots,h_t-1\}$ we have that $\frac{\mu(t,i+1)}{\mu(t,i)}\geq \alpha$. It follows from (\ref{vbn22bxvcn0}) that \begin{equation}\label{rtvv265v}\alpha^{h_t-1}\leq h_t\leq \vert t\vert\mbox{.}\end{equation} 
%Let $\Psi(n)=\Sigma\cap\bigcup_{i\geq n}\Alphabet^i$
Let $c=\frac{1}{\ln{\alpha}}\in \mathbb{R}^+$. Then $\vert t\vert=\alpha^{c\ln{\vert t\vert}}$. Since $h_t\geq \phi(\vert t\vert)$ we get that \begin{equation}\label{yru77bv46}\frac{\alpha^{h_t-1}}{\vert t\vert}\geq\frac{\alpha^{\phi(\vert t\vert)-1}}{\vert t\vert}=\frac{\alpha^{\phi(\vert t\vert)-1}}{\alpha^{c\ln{\vert t\vert}}}=\alpha^{\phi(\vert t\vert)-1-c\ln{\vert t\vert}}\mbox{.}\end{equation}
Because $\lim_{n\rightarrow\infty}\left(\phi(n)-\ln{n}\right)=\infty$ the equation (\ref{yru77bv46}) implies that there is $n_0$ such that for each $t\in \Sigma$ with $\vert t\vert>n_0$ we have that \begin{equation}\label{uub2vx00l9}\frac{\alpha^{h_t-1}}{\vert t\vert}\geq\alpha^{\phi(\vert t\vert)-1-c\ln{\vert t\vert}}>1\mbox{.}\end{equation}
From (\ref{rtvv265v}) and (\ref{uub2vx00l9}) we have that $\alpha^{h_t-1}\leq \vert t\vert$ and $\frac{\alpha^{h_t-1}}{\vert t\vert}>1$, which is a contradiction. We conclude there is no such $\alpha$. In consequence, for each $\beta\in \mathbb{R}^+$ with $\beta>1$ there is $t\in \Sigma$ and $i\in \{1,2,\dots,h_t-1\}$ such that $\frac{\mu(t,i+1)}{\mu(t,i)}\leq \beta$.

Let $j\in \mathbb{N}$, let \begin{equation}\label{tri8iu44nz}\gamma\leq\frac{1}{j}+1\in\mathbb{R}^+\mbox{,}\end{equation} let  $t\in \Sigma$, and $i\in \{1,2,\dots,h_t\}$ be such that $\frac{\mu(t,i+1)}{\mu(t,i)}\leq \gamma$. Let $\delta=\frac{\mu(t,i+1)}{\mu(t,i)}\leq \gamma$. Let $u,v\in \Prefix(t)$ be such that $\vert u\vert=\mu(t,i)$ and $\vert v\vert=\mu(t,i+1)$. Then $v$ is a periodic palindrome with a period $\vert v\vert -\vert u\vert=\mu(t,i+1)-\mu(t,i)=\mu(t,i)\delta-\mu(t,i)=\mu(t,i)(\delta -1)$; see Lemma \ref{id8ieubmzmfj}. 
Lemma \ref{tudjkdi8545} implies that there are $a\in\Pal$ and $b\in\Pal^+$ such that $(ab)^{k}a=v$ for some $k\in \mathbb{N}$. From Lemma \ref{tudjkdi8545} we have also that $\vert ab\vert$ is the period of $v$. Thus \begin{equation}\label{8tut8gugbb}\vert ab\vert=\mu(t,i)(\delta-1)\leq \mu(t,i)(\gamma-1)\mbox{.}\end{equation}
From (\ref{tri8iu44nz}) and (\ref{8tut8gugbb}) it follows that
\begin{equation}\label{uejehbn3}\vert ab\vert\leq \mu(t,i)(\gamma-1)\leq \mu(t,i)\frac{1}{j}\mbox{.}\end{equation}
Note that $v=(ab)^ka$ and $u\in \Prefix((ab)^k)$. Since $\mu(t,i)=\vert u\vert$ we get that $\frac{\mu(t,i)}{\vert ab\vert}\leq k$. From (\ref{uejehbn3}) we have that 
\[j\leq \frac{\mu(t,i)}{\vert ab\vert}\leq k\mbox{.}\]
Thus for arbitrary $j\in \mathbb{N}$ we found $t,a,b,k$ such that $(ab)^k\in \Prefix(t)$ and $j\leq k$.
The proposition follows.
\end{proof}

A corollary of Proposition \ref{uyzm3m2mzo} and Proposition \ref{ry73u39udh} says that if $w$ is an infinite word with a bounded palindromic length then for each positive integer $j$ there are palindromes $a,b$ such that $(ab)^j$ is a factor of $w$ and $ab$ is a nonempty word.
\begin{corollary}
If $w\in \Alphabet^{\mathbb{N}}$, $k\in \mathbb{N}$, and $k\geq \max\{\PL(t)\mid t\in \Factor(w)\}$ then for each $j\in \mathbb{N}$ there are $a\in \Pal$ and $b\in \Pal^+$ such that $(ab)^j\in  \Factor(w)$.  
\end{corollary}
\begin{proof}
Just take $\Sigma=\Omega(w,k)$. Obviously $\lim_{n\rightarrow \infty}\left(\tau(n,k)-\ln{n}\right)=\infty$. Then Proposition \ref{ry73u39udh} implies the corollary. 
\end{proof}

\section{Palindromic length of concatenation}
In this section we present some known results about the palindromic length of concatenation of two words.

The first lemma shows the very basic property of the palindromic length that the palindromic length of concatenation of two words $x$ and $y$ is lower than or equal to the sum of palindromic length of $x$ and $y$.
\begin{lemma}
\label{r6yrt63yyt3t}
If $x,y\in \Alphabet^*$ then $\PL(xy)\leq \PL(x)+\PL(y)$.
\end{lemma}
\begin{proof}
If $x=\epsilon$ or $y=\epsilon$ then obviously $\PL(xy)=\PL(x)+\PL(y)$. Hence suppose that $x,y\in \Alphabet^+$. 
Let $i=\PL(x)$ and $j=\PL(y)$. Let $(t_1,t_2,\dots,t_i)\in \MPF(x)$ and $(u_1,u_2,\dots,u_j)\in \MPF(y)$. Then $t_1t_2\dots t_iu_1u_2\dots u_j$ is a factorization of $xy$ into $i+j$ palindromes. Consequently $\PL(xy)\leq i+j=\PL(x)+\PL(y)$. This completes the proof.
\end{proof}

An another basic property of the palindromic length says that if \[(t_1,t_2,\dots, t_k)\in\MPF(w)\] is a minimal palindromic factorization of the word $w$ then the palindromic length of the factor $t_it_{i+1}\dots t_j$ is equal to $j-i+1$ for each $i,j\in \{1,2,\dots,k\}$ and $i\leq j$. 
\begin{lemma}
\label{r6nb2b2b2b2}
If $w\in \Alphabet^+$, $k=\PL(w)$, and $(t_1,t_2,\dots, t_k)\in \MPF(w)$ then for each $i,j\in \{1,2,\dots,k\}$ with $i\leq j$ we have that $\PL(t_it_{i+1}\dots t_j)=j-i+1$.
\end{lemma}
\begin{proof}
Since the word $t_it_{i+1}\dots t_j$ is concatenated of $j-i+1$ palindromes it is clear that $\PL(t_it_{i+1}d\dots t_j)\leq j-i+1$. Suppose that $\PL(t_it_{i+1}d\dots t_j)=m<j-i+1$. It would follow from Lemma \ref{r6yrt63yyt3t} that \[\begin{split}\PL(t_1t_2\dots t_k)\leq \PL(t_1t_2\dots t_{i-1})+\PL(t_it_{i+1}d\dots t_j) + \PL(t_{j+1}t_{j+2}\dots t_k)\leq \\ i-1+m+k-j<i-1 + j-i+1+k-j=k\mbox{.}\end{split}\]
This is contradiction, since $\PL(t_1t_2\dots t_k)=k$. The lemma follows.
\end{proof}

The following result has been proved in \cite{10.1007/978-3-319-66396-8_19}. It says that if $x,y$ are words then the palindromic length of $y$ is the maximal absolute difference of palindromic lengths of $x$ and $xy$; i.e. $\lvert\PL(x)-\PL(xy)\rvert\leq \PL(y)$.

\begin{lemma}(see \cite[Lemma $6$]{10.1007/978-3-319-66396-8_19})
\label{je8eurikk34iu}
If $x,y\in \Alphabet^*$ then \begin{itemize}
\item
$\PL(y)\leq \PL(x)+\PL(xy)$ and
\item
$\PL(x)\leq \PL(y)+\PL(xy)$.
\end{itemize}
\end{lemma}

The immediate corollary of Lemma \ref{je8eurikk34iu} is that if $x$ is a word and $y$ is a palindrome then the absolute difference of palindromic lengths of $x$ and $xy$ is at most $1$.
\begin{corollary}
\label{s8sjhhu99e}
If $x,y\in \Alphabet^*$ and $y\in \Pal$ then $\lvert\PL(xy)-\PL(x)\rvert\leq 1$.
\end{corollary}
\begin{proof}
It is enough to consider $y$ in Lemma \ref{je8eurikk34iu} to be a palindrome. Thus we have $\PL(y)=1$ if $y\not=\epsilon$ or $\PL(y)=0$ if $y=\epsilon$. The corollary follows.
\end{proof}

The next simple Corollary of Lemma \ref{je8eurikk34iu} says that if $x,y$ are words such that $xy$ is a palindrome then the absolute difference in palindromic lengths of $x$ and $y$ is at most $1$.
\begin{corollary}
\label{t8e9kffie8i}
If $x,y\in \Alphabet^*$ and $xy\in \Pal$ then $\lvert \PL(x)-\PL(y)\rvert\leq 1$.
\end{corollary}
\begin{proof}
If $x=y^R$ then  $\PL(x)-\PL(y)=0$, because clearly $\PL(y)=\PL(y^R)$.
Suppose that $x\not=y^R$. It follows that $\vert x\vert \not=\vert y\vert$, since $xy\in \Pal$. Without loss of generality suppose that $\vert x\vert>\vert y\vert$. Let $\bar x$ be such that $x=y^R\bar x$. Then $xy=y^R\bar xy$. Thus $\bar x\in \Pal^+$. Corollary \ref{s8sjhhu99e} implies that $\vert\PL(y^R\bar x)-\PL(y)\rvert\leq 1$. The corollary follows.
\end{proof}

\section{Concatenation of periodic palindromes}

To simplify the notation of the next two lemmas and the theorem we define an auxiliary set $\Delta$.
Let $\Delta$ be the set of all $4$-tuples $(u,d,v,n)$ such that 
\begin{itemize}
\item $d\in \Pal^+$,
\item $v\in \Pal$,
\item $u\in \Suffix(d)\setminus\{\epsilon\}$,
\item $n\in \mathbb{N}$,
\item $\vert dv\vert=\MinPer(dvd)$, and
\item $n\geq 3\PL(u)$.
\end{itemize}

\begin{remark}
The set $\Delta$ contains all $4$-tuples $(u,v,d,n)$ such that $d$ is a nonempty palindrome, $v$ is a palindrome (possibly empty), $u$ is a nonempty suffix of $d$, $\vert dv\vert$ is the minimal period of the word $dvd$, and $n$ is a positive integer such that $n\geq 3\PL(u)$. It follows that $n\geq 3$, since $u$ is nonempty and thus $\PL(u)\geq 1$. 
\end{remark}

\begin{lemma}
\label{11scnfk9p}
If $(u,v,d,n)\in \Delta$, $r\in\Factor(u(vd)^n)$, and $\vert r\vert\geq 3\vert vd\vert$ then $dvd\in \Factor(r)$.
\end{lemma}
\begin{proof}
Let $\bar w=u(vd)^n$, let $p\in \Prefix(r)$ with $\vert p\vert=3\vert vd\vert$, and let $\bar i,\bar j\in \{1,2,\dots,\vert \bar w\vert\}$ be such that $p=\bar w[\bar i,\bar j]$. Let $\bar u\in \Prefix(d)$ be such that $d=\bar uu$. Note that $\vert uv\bar u\vert=\vert vd\vert$ and thus $(uv\bar u, \beta)\in \Period(\bar w)$, where $\beta=\frac{\vert \bar w\vert}{\vert uv\bar u\vert}>1$. 

Let $k\in \mathbb{N}_0$ and $w\in \Suffix(\bar w)$ be such that $\bar w=(uv\bar u)^k w$, $\bar i>\vert (uv\bar u)^k\vert$, and $\bar i\leq \vert (uv\bar u)^{k+1}\vert$. Obviously such $k$ and $w$ exist. Let $i=\bar i-k\vert uv\bar u\vert$ and $j=\bar j-k\vert uv\bar u\vert$. It is easy to see that $p=w[i,j]$. 

We distinguish:
\begin{itemize}
\item If $i\in \{1,2,\dots,\vert u\vert\}$ then $p=tvdvdv\bar t$ for some $t\in\Suffix(u)$ and for $\bar t$ such that $d=\bar tt$.
\item If $i\in \{\vert u\vert+1, \vert u\vert+2, \dots, \vert uv\vert\}$ then $p=tdvdvd\bar t$ for some $t\in\Suffix(v)$ and for $\bar t$ such that $v=\bar tt$.
\item If $i\in \{\vert uv\vert +1, \vert uv\vert +2, \dots, \vert uv\vert +\vert \bar u\vert\}$ then $p=tvdvdv\bar t$ for some $t\in\Suffix(d)$ and for $\bar t$ such that $d=\bar tt$.
\end{itemize}
In all three cases one can see that $dvd\in \Factor(p)$. It is easy to see that if $dvd\in \Factor(p)$ then $dvd\in \Factor(r)$ for each $r\in \Factor(w)$ with $p\in \Prefix(r)$. The lemma follows.
\end{proof}
\begin{remark}
Note in the previous proof that with the condition $\vert r\vert\geq \vert (vd)^2\vert$ it would be possible that $dvd\not\in \Factor(p)$. In the cases $1$ and $3$ we would have $p=tvdv\bar t$. That is why the condition $\vert r\vert\geq \vert (vd)^3\vert$ necessary is. For this reason in the definition of $\Delta$ we state that $n\geq 3\PL(u)$.
\end{remark}
The next lemma shows that if $(u,v,d,n)\in\Delta$, $k$ is the palindromic length of $u$, and $(t_1,t_2,\dots,t_k)\in\MPF(u(vd)^n)$ is a minimal palindromic factorization of $u(vd)^n$ then there is $j\in\{1,2,\dots,k\}$ such that $t_j$ is a palindrome having the factor $dvd$ in the ``center'' of $t_j$; formally $t_j=pd(vd)^{\gamma}p^R$ for some positive integer $\gamma$ and for some proper suffix $p$ of $dv$.
\begin{lemma}
\label{rhujdu333dhj}
If $(u,v,d,n)\in \Delta$, $w=u(vd)^n$, $k=\PL(w)$, and \[(t_1,t_2,\dots, t_k)\in \MPF(w)\] then there are $j\in \{1,2\dots, k\}$, $p\in \Suffix(dv)\setminus\{dv\}$, and $\gamma\in \mathbb{N}$ such that $t_j=pd(vd)^{\gamma}p^R$.
\end{lemma}
\begin{proof}
Suppose that $\vert t_i\vert<3\vert vd\vert$ for each $i\in \{1,2,\dots,k\}$. It follows that \[\vert t_1t_2\dots t_k\vert< 3k\vert vd \vert\mbox{.}\] Since $u(vd)^n=t_1t_2\dots t_k$ and $n\geq 3k\geq 3$ it is a contradiction. 
It follows that there is $j$ such that $\vert t_j\vert \geq\vert (vd)^3\vert$. Lemma \ref{11scnfk9p} asserts that $dvd\in \Factor(t_j)$. Then clearly there are $\gamma\in \mathbb{N}$ and $p_1,p_2\in \Alphabet^*$ such that $p_1\in \Suffix(dv)\setminus\{dv\}$, $p_2\in \Prefix(vd)\setminus\{vd\}$, and $t_j=p_1d(vd)^{\gamma}p_2$.

To get a contradiction suppose that $p_1\not=p_2^R$.
Without loss of generality suppose that $\vert p_1\vert >\vert p_2\vert$. It follows that $p_2\in \Prefix(p_1^R)$. Obviously $p_1d(vd)^{\gamma}p_1^R\in \Pal$. Thus we have two palindromes $p_1d(vd)^{\gamma}p_1^R$ and $p_1d(vd)^{\gamma}p_2$. Lemma \ref{id8ieubmzmfj} implies that $p_1d(vd)^{\gamma}p_1^R$ is periodic with a period \[\delta=\vert p_1d(vd)^{\gamma}p_1^R\vert-\vert p_1d(vd)^{\gamma}p_2\vert =\vert p_1\vert -\vert p_2\vert\mbox{.}\]
Clearly $\delta<\vert dv\vert$. This is a contradiction to the condition $\vert dv\vert =\MinPer(dvd)$, see Definition of $\Delta$. We conclude that $p_1=p_2^R$.
The lemma follows.
\end{proof}

The main theorem of the article says that if $v,d$ are palindromes, $d$ is nonempty, $u$ is a nonempty suffix of $d$, $k=\PL(u)$, $\vert dv\vert$ is the minimal period of $dvd$, and $n$ is a positive integer such that $n\geq 3k$ then the palindromic length of the word $u(vd)^n$ is bigger than or equals to the palindromic length of $u$.
\begin{theorem}
If $(u,v,d,n)\in \Delta$, $k=\PL(u)$, and $w=u(vd)^n$ then $\PL(w)\geq k$.
\end{theorem}
\begin{proof}
Let $(t_1,t_2,\dots, t_k)\in \MPF(w)$. Lemma \ref{rhujdu333dhj} asserts that there are $j\in \{1,2\dots, k\}$, $p\in \Suffix(dv)\setminus\{dv\}$, and $\gamma\in \mathbb{N}$ such that $t_j=pd(vd)^{\gamma}p^R$.

Let $a\in \Prefix(w)$ and $b\in \Suffix(w)$ be such that
$w=at_jb$. Realize that $a=t_1t_2\dots t_{j-1}$ and $b=t_{j+1}t_{j+2}\dots t_k$. Note that $a$ or $b$ can be the empty word; then $j=1$ or $j=k$ respectively.
Lemma \ref{r6nb2b2b2b2} implies that \begin{equation}\label{rt6dydtyee22}\begin{split}\PL(w)= \PL(t_1t_2\dots t_{j-1})+\PL(t_j)+\PL(t_{j+1}t_{j+2}\dots t_k)= \\ \PL(a)+\PL(t_j)+\PL(b)\mbox{.}\end{split}\end{equation}

We distinguish three distinct cases.
\begin{enumerate}
\item \label{rutikj887e8}
$u\not\in \Prefix(a)$: This case is depicted in Table \ref{Fig_oiwieofk}. Let $u_2\in \Suffix(u)$ be such that $u=au_2$. Let $\bar p\in \Suffix(d)$ be such that $\bar pu_2=d$. It follows  that $u_2^R\bar p^R=d$ and $p^R\bar p^R=vd$. 

Then we have that $u_2^Rb=u_2^R\bar p^R(vd)^{\beta}=d(vd)^{\beta}\in \Pal^+$ for some $\beta\in \mathbb{N}_0$. Hence $\PL(u_2^R\bar p^R(vd)^{\beta})=1$. In consequence $\PL(u_2)\geq \PL(b)-1$ and \begin{equation}\label{tthnb87r8}\PL(b)\geq\PL(u_2)-1\mbox{,}\end{equation} since $\PL(u_2^R)=\PL(u_2)$ and $u_2^Rb\in \Pal^+$; see Corollary \ref{t8e9kffie8i}. 

Lemma \ref{r6yrt63yyt3t} implies that \begin{equation}\label{55tn55f8}\PL(a)+\PL(u_2)\geq \PL(u)\mbox{.}\end{equation}

From (\ref{rt6dydtyee22}), (\ref{tthnb87r8}), and (\ref{55tn55f8}) we have that \[\begin{split}\PL(w)=\PL(a)+\PL(t_j)+\PL(b)\geq \PL(a)+1+\PL(u_2)-1\geq \PL(u)\mbox{.}\end{split}\] 

\begin{table}[h]
\centering
\begin{tabular}{|c|c|l|l|c|l|c|l|}
\hline
$a$   & \multicolumn{5}{c|}{$t_j$}                                                                                                 & \multicolumn{2}{c|}{$b$}                \\ \hline
$a$      & \multicolumn{2}{c|}{$p$}                              &  $d(vd)^{\gamma}$         & \multicolumn{2}{c|}{$p^R$} & \multicolumn{1}{l|}{$\bar p^R$} &  $(vd)^{\beta}$               \\ \hline
$a$ & \multicolumn{1}{l|}{$u_2$} & \multicolumn{1}{c|}{$v$} & \multicolumn{1}{c|}{$d(vd)^{\gamma}$} & $v$        & $u_2^R$       & $\bar p^R$            & $(vd)^{\beta}$ \\ \hline
\multicolumn{2}{|c|}{$u$}          & \multicolumn{2}{c|}{$(vd)^{\gamma+1}$}                               & \multicolumn{1}{c|}{$v$}       & \multicolumn{2}{c|}{$d$}                 &    $(vd)^{\beta}$             \\ \hline
\end{tabular}
\caption{Case \ref{rutikj887e8}: The structure of the word $w$ with $u\not\in \Prefix(a)$.}
\label{Fig_oiwieofk}
\end{table}
\item 
\label{t8id22nbx}
$u\in \Prefix(a)$ and $p\in \Suffix(v)$: This case is depicted in Table \ref{Fig_38hhyegy2}.
Let $\bar p\in \Prefix(v)$ be such that
$\bar pp=v$. Note that if $p=v$ then $\bar p=\epsilon$, and if $p=\epsilon$ then $\bar p=v$. 
It is easy to verify that $b=\bar p^Rd(vd)^{\beta}$ for some $\beta\in \mathbb{N}_0$ and $a=u(vd)^{\alpha}\bar p$ for some $\alpha\in \mathbb{N}_0$. 

Let $\bar a$ be such that $a=u\bar a$.
We have that $\bar a=(vd)^{\alpha}\bar p$ and $b=\bar p^Rd(vd)^{\beta}$. It follows that either $\bar a=b^Rd(vd)^{\delta}$ or $b=\bar a^Rd(vd)^{\delta}$ for some $\delta\in \mathbb{N}_0$.

Since $d(vd)^{\delta}\in \Pal$, Corollary \ref{s8sjhhu99e} implies that \begin{equation}\label{rubbnxb2bb}\lvert\PL(\bar a)-\PL(b)\rvert\leq 1\mbox{.}\end{equation} It follows from Lemma \ref{je8eurikk34iu} that $\PL(a)+\PL(\bar a)\geq \PL(u)$ and in consequence \begin{equation}\label{6rytytf88d}\PL(a)\geq \PL(u)-\PL(\bar a)\mbox{.}\end{equation}
From (\ref{rt6dydtyee22}), (\ref{rubbnxb2bb}), and (\ref{6rytytf88d}) we have that
\[\begin{split}\PL(w)=\PL(a)+\PL(t_j)+\PL(b)\geq \PL(u)-\PL(\bar a) + 1 + \PL(b) \geq \\ \PL(u)-\PL(\bar a) + 1 + \PL(\bar a)-1=\PL(u)\mbox{.}\end{split}\]
\begin{table}[h]
\centering
\begin{tabular}{|c|l|l|c|l|c|c|l|}
\hline
\multicolumn{3}{|c|}{$a$}                                     & \multicolumn{3}{c|}{$t_j$}                                            & \multicolumn{2}{c|}{$b$}             \\ \hline
$u$                    & $(vd)^{\alpha}$ & $\bar p$           & $p$           & \multicolumn{1}{c|}{$d(vd)^{\gamma}$} & $p^R$         & $\bar p^R$         & $d(vd)^{\beta}$ \\ \hline
\multicolumn{1}{|l|}{} &                  \multicolumn{2}{c|}{$\bar a$} & \multicolumn{2}{c|}{}         & \multicolumn{2}{c|}{$v$} &                 \\ \hline
\end{tabular}
\caption{Case \ref{t8id22nbx}: The structure of the word $w$ with $u\in \Prefix(a)$ and $p\in \Suffix(v)$.}
\label{Fig_38hhyegy2}
\end{table}
\item \label{rru8ub121q} $u\in \Prefix(a)$ and $p\not \in \Suffix(v)$: This case is depicted in Table \ref{Fig_d9eikji344hb}. Since $p\in \Suffix(vd)\setminus\{vd\}$ and $p\not\in \Suffix(v)$ it follows that $p\in \Suffix(dv)\setminus(\Suffix(v)\cup\{dv\})$. 

Let $\bar p\in \Prefix(d)$ be such that $\bar pp=dv$ and consequently $p^R\bar p^R=vd$.
Then $a=u(vd)^{\alpha}\bar p$ for some $\alpha\in \mathbb{N}_0$  and $b=\bar p(vd)^{\beta}$ for some $\beta\in \mathbb{N}_0$.

Let $\bar a$ be such that $a=u\bar a$.
We have that $\bar a=v(dv)^{\alpha}\bar p$. It follows that either $\bar a=b^R(vd)^{\delta}v$ or $b=\bar a^R(vd)^{\delta}v$ for some $\delta\in \mathbb{N}_0$.

The rest of the proof of Case \ref{rru8ub121q} is analogue to Case \ref{t8id22nbx}:
Since $v(dv)^{\delta}\in \Pal$, Corollary \ref{s8sjhhu99e} implies that \begin{equation}\label{oou39mdhn3}\lvert\PL(\bar a)-\PL(b)\rvert\leq 1\mbox{.}\end{equation} It follows from Lemma \ref{je8eurikk34iu} that $\PL(a)+\PL(\bar a)\geq \PL(u)$ and in consequence \begin{equation}\label{irki999f8hb}\PL(a)\geq \PL(u)-\PL(\bar a)\mbox{.}\end{equation}
From (\ref{rt6dydtyee22}), (\ref{oou39mdhn3}), and (\ref{irki999f8hb}) we have that
\[\begin{split}\PL(w)=\PL(a)+\PL(t_j)+\PL(b)\geq \PL(u)-\PL(\bar a) + 1 + \PL(b) \geq \\ \PL(u)-\PL(\bar a) + 1 + \PL(\bar a)-1=\PL(u)\mbox{.}\end{split}\]
\begin{table}[h]
\centering
\begin{tabular}{|c|l|l|c|l|c|c|l|}
\hline
\multicolumn{3}{|c|}{$a$}                                     & \multicolumn{3}{c|}{$t_j$}                                            & \multicolumn{2}{c|}{$b$}             \\ \hline
$u$                    & $v(dv)^{\alpha}$ & $\bar p$           & $p$           & \multicolumn{1}{c|}{$d(vd)^{\gamma}$} & $p^R$         & $\bar p^R$         & $(vd)^{\beta}$ \\ \hline
\multicolumn{1}{|l|}{}                  & \multicolumn{2}{c|}{$\bar a$} &  \multicolumn{2}{c|}{}           & \multicolumn{2}{c|}{$vd$} &                 \\ \hline
\end{tabular}
\caption{Case \ref{rru8ub121q}: The structure of the word $w$ with $u\in \Prefix(a)$ and $p\not \in \Suffix(v)$.}
\label{Fig_d9eikji344hb}
\end{table}
\end{enumerate}
We proved for each case that $\PL(w)\geq \PL(u)$. Since obviously for each $u$ and each $p$ one of the three cases applies, this completes the proof.
\end{proof}

%%%%%%%%%%%%%%%%%%%%%%%%%%%%%%%%%%%%%%%%%%%

\section*{Acknowledgments}
This work was supported by the Grant Agency of the Czech Technical University in Prague, grant No. SGS20/183/OHK4/3T/14.

\bibliographystyle{siam}
\IfFileExists{biblio.bib}{\bibliography{biblio}}{\bibliography{../!bibliography/biblio}}

\end{document}